\documentclass[11pt,a4paper]{article}
\usepackage[margin=1in]{geometry}
\usepackage{amsmath,amssymb,amsthm,booktabs,color,doi,enumitem,graphicx,url}
\usepackage{booktabs,latexsym}
\usepackage{thmtools}
\usepackage{thm-restate}
\usepackage[basic]{complexity}
\usepackage{stmaryrd}
\usepackage[numbers,sort&compress]{natbib}
\usepackage{microtype}
\usepackage{hyperref}
\usepackage{multirow}
\usepackage{comment}
\usepackage{algorithm}
\usepackage{algpseudocode}

\newtheorem{theorem}{Theorem}
\newtheorem{corollary}[theorem]{Corollary}
\newtheorem{lemma}[theorem]{Lemma}

\newtheorem{fact}{Fact}

\theoremstyle{definition}
\newtheorem{definition}{Definition}

\theoremstyle{remark}
\newtheorem{remark}{Remark}

\DeclareMathOperator{\id}{id}
\DeclareMathOperator{\mc}{mc}
\DeclareMathOperator{\sol}{sol}

\newclass{\bs}{bs}
\renewclass{\C}{C}
\newclass{\D}{D}
\newclass{\local}{LOCAL}
\newclass{\lca}{CentLOCAL}
\newclass{\partree}{ParallelDecTree}
\newclass{\nice}{NICE}

\DeclareMathOperator{\F}{\mathcal{F}}
\newcommand{\A}{\mathcal{A}}

\newcommand{\T}{\mathcal{T}}
\newcommand{\G}{\mathcal{G}}
\renewcommand{\H}{\mathcal{H}}

\newcommand{\mis}{\text{\textsc{mis}}}

\newcommand{\mm}{\text{\textsc{mm}}}

\newcommand{\dcolor}{\text{$(\Delta+1)$-\textsc{colour}}}
\newcommand{\mcm}{\text{\textsc{mcm}}}
\newcommand{\eps}{\varepsilon}
\newcommand{\mwm}{\text{\textsc{mwm}}}
\newcommand{\simu}{\text{\textsc{sim}}}
\newcommand{\simur}{\sim_{\text u.a.r.}}
\newcommand{\oracleG}{$\mathcal{O}^G$}
\newcommand{\oracleGm}{\mathcal{O}^G}

\newcommand{\simuA}{$\textsf{sim}^A$}

\hypersetup{
    colorlinks=true,
    linkcolor=black,
    citecolor=black,
    filecolor=black,
    urlcolor=[rgb]{0,0.1,0.5},
    pdfauthor={Mika G\"o\"os, Juho Hirvonen, Reut Levi, Moti Medina, Jukka Suomela},
    pdftitle={Non-Local Probes Do Not Help with Graph Problems},
}

\newenvironment{myabstract}
               {\list{}{\listparindent 1.5em%
                        \itemindent    \listparindent
                        \leftmargin    0pt
                        \rightmargin   0pt
                        \parsep        0pt}%
                \item\relax}
               {\endlist}

\newenvironment{mycover}
               {\list{}{\listparindent 0pt
                        \itemindent    \listparindent
                        \leftmargin    0pt
                        \rightmargin   0pt
                        \parsep        0pt}%
                \raggedright
                \item\relax}
               {\endlist}

\begin{document}

\vspace*{2ex}
\begin{mycover}
{\huge\bfseries Non-Local Probes Do Not Help\\with Graph Problems\par}
\bigskip
\bigskip

\textbf{Mika G\"o\"os}

\nolinkurl{mika.goos@mail.utoronto.ca}
\medskip

{\small Department of Computer Science, University of Toronto, Canada \par}
\bigskip

\textbf{Juho Hirvonen}

\nolinkurl{juho.hirvonen@aalto.fi}
\medskip

{\small Helsinki Institute for Information Technology HIIT, \\ Department of Computer Science, Aalto University, Finland\par}
\bigskip

\textbf{Reut Levi}

\nolinkurl{rlevi@mpi-inf.mpg.de}
\medskip

{\small MPI for informatics, \\ Saarbr\"{u}cken 66123, Germany \par}
\bigskip

\textbf{Moti Medina}

\nolinkurl{mmedina@mpi-inf.mpg.de}
\medskip

{\small MPI for informatics, \\ Saarbr\"{u}cken 66123, Germany \par}
\bigskip

\textbf{Jukka Suomela}

\nolinkurl{jukka.suomela@aalto.fi}
\medskip

{\small Helsinki Institute for Information Technology HIIT, \\ Department of Computer Science, Aalto University, Finland\par}

\end{mycover}
\bigskip
\begin{myabstract}
\noindent\textbf{Abstract.}
This work bridges the gap between distributed and centralised models of computing in the context of sublinear-time graph algorithms. A priori, typical centralised models of computing (e.g., parallel decision trees or centralised local algorithms) seem to be much more powerful than distributed message-passing algorithms: centralised algorithms can directly probe any part of the input, while in distributed algorithms nodes can only communicate with their immediate neighbours. We show that for a large class of graph problems, this extra freedom does not help centralised algorithms at all: for example, efficient stateless deterministic centralised local algorithms can be simulated with efficient distributed message-passing algorithms. In particular, this enables us to transfer existing lower bound results from distributed algorithms to centralised local algorithms.
\end{myabstract}

\thispagestyle{empty}
\setcounter{page}{0}
\newpage


\section{Introduction}

A lot of recent work on efficient graph algorithms for massive graphs can be broadly classified in one of the following categories:
\begin{enumerate}
    \item \emph{Probe-query models} \cite{rubinfeld11fast,even14deterministic,even14best,even15distributed,mansour13local,levi14local,alon12space-efficient,alon10constant}: Here typical applications are related to \emph{large-scale network analysis}: we have a huge storage system in which the input graph is stored, and a computer that can access the storage system. The user of the computer can make \emph{queries} related to the properties of the graph.

    Conceptually, we have two separate entities: the input graph and a computer. Initially, the computer is unaware of the structure of the graph, but it can \emph{probe} it to learn more about the structure of the graph. Typically, the goal is to answer queries after a \emph{sublinear number of probes}.
    \item \emph{Message-passing models} \cite{linial92locality,naor95what,lynch96book,peleg00distributed,kuhn04what,suomela13survey,lenzen2008leveraging,czygrinow2008fast}: In message-passing models, typical applications are related to \emph{controlling large computer networks}: we have a computer network (say, the Internet) that consists of a large number of network devices, and the devices need to collaborate to solve a graph problem related to the structure of the network so that each node knows \emph{its own part of the solution} when the algorithm stops.

    Conceptually, each node of the input graph is a computational entity. Initially, the nodes are only aware of their own identity and the connections to their immediate neighbours, but the nodes can \emph{exchange messages} with their neighbours in order to learn more about the structure of the graph. Typically, the goal is to solve graph problems in a \emph{sublinear number of communication rounds}.
\end{enumerate}

\subsection{Example: Vertex Colouring}

Using the task of finding a proper vertex colouring as an example, the external behaviour of the algorithms can be described as follows:
\begin{enumerate}
    \item \emph{Probe-query models:} The user can make queries of the form ``what is the colour of node $v$?'' The answers have to be consistent with some fixed feasible solution: for example, if we query the same node twice, the answer has to be the same, and if we query two adjacent nodes, their colours have to be different.
    \item \emph{Message-passing models:} The local output of node $v$ is the colour of node $v$. The local outputs constitute some feasible solution: the local outputs of two adjacent nodes have to be different.
\end{enumerate}

\subsection{Message-Passing Models and Locality}

From our perspective, the key difference between probe-query models and message-passing models is that the structure of the graph constrains the behaviour of message-passing algorithms, but not probe-query algorithms:
\begin{enumerate}
    \item \emph{Probe-query models:} Given any query $v$, the algorithm is free to probe any parts of the input graph. In particular, it does not need to probe node $v$ or its immediate neighbours.
    \item \emph{Message-passing models:} Nodes can only exchange messages with their immediate neighbours. For example, in $1$ communication round, a node can only learn information related to its immediate neighbours. More generally, after $t$ communication rounds, in any message-passing algorithm, each node can only be aware of information that was initially available in its radius-$t$ neighbourhood.
\end{enumerate}
In essence, efficient message-passing algorithms have a high \emph{locality}: if the running time of a message-passing algorithm is $t$, then the local output of a node $v$ can only depend on the information that was available within its radius-$t$ neighbourhood in the input graph.

\subsection{Trivial: from Message-Passing to Probe-Query}

One consequence of locality is that we can fairly easily \emph{simulate efficient message-passing algorithm in probe-query models} (at least for deterministic algorithms).

If we have a message-passing algorithm $A$ with a running time of $t$ for, say, graph colouring, we can turn it easily into a probe-query algorithm $A'$ for the same problem: to answer a query $v$, algorithm $A'$ simply gathers the radius-$t$ neighbourhood of node $v$ and simulates the behaviour of $A$ in this neighbourhood.

In particular, if $t$ is a constant, and the maximum degree of the input graph is bounded by a constant, the probe complexity of $A'$ is also bounded by a constant.

\subsection{Impossible: from Probe-Query to Message-Passing}

At first, it would seem that the converse cannot hold: probe-query algorithms can freely probe remote parts of the network, and therefore they cannot be simulated with efficient message-passing algorithms.

Indeed, it is easy to construct \emph{artificial} graph problems that are trivial to solve in probe-query models in constant time and that take linear time to solve in the message-passing models. Consensus-like problems provide a simple example.

In the \emph{binary consensus} problem, the nodes of the network are labeled with inputs $0$ and $1$, all nodes have to produce the same output, and the common output has to be equal to the input of at least one node. In a probe-query algorithm, we can simply always follow the local input of node number $1$: regardless of the query $v$, we will probe node $1$, check what was its local input, and answer accordingly. It is straightforward to see that any message-passing algorithm requires linear time (consider three cases: a path with inputs $00\ldots0$, a path with with inputs $11\ldots1$, and a path with inputs $00\ldots011\ldots1$).

\subsection{Our Contribution: Remote Probes Are of Little Use}

While the consensus problem demonstrates that the possibility to probe remote parts of the network makes probe-query models strictly stronger than message-passing models, there seem to be few \emph{natural} graph problems in which remote probes would help.

In this work, we formalise this intuition. We show that for a large class of graph problems, remote probes do not give probe-query algorithms any advantage over message-passing algorithms.

Among others, we will show that for a large family of problems which also includes the so-called \emph{locally checkable problems} (LCL) \cite{naor95what}, probe-query algorithms in Rubinfeld et al.'s model~\cite{rubinfeld11fast} can be efficiently simulated with message-passing algorithms in Linial's model~\cite{linial92locality}.
\subsection{Corollary: Probe-Query Lower Bounds}

While lower-bound results in probe-query models are scarce, there is a lot of prior work on lower-bound results in message-passing models. Indeed, the very concept of locality makes it relatively easy to derive lower-bound results for message-passing models: a problem cannot be solved in time $t$ if there is at least one graph in which the output of some node necessarily depends on information that is not available in its radius-$t$ neighbourhood.

Our simulation result makes it now possible to take existing lower bounds for message-passing models and use them to derive analogous lower bounds for probe-query models.

\subsection{Related Work}

While both message-passing models and probe-query models have been studied extensively, it seems that there is little prior work on their connection. The closest prior work that we are aware of is the 1990 paper by Fich and Ramachandran~\cite{fich90linkedlists}. They show a result similar to our main theorem for one graph problem---graph colouring in cycles---with problem-specific arguments. While our techniques are different, our main result can be seen as a generalisation of their work from a single graph problem to a broad family of graph problems.

On the side of probe-query algorithms, our main focus is on the $\lca$ model~\cite{rubinfeld11fast}. Algorithms in the $\lca$ model are abundant, and include algorithms for various graph problems such as \emph{maximal matching}, \emph{approximated maximum matching}, \emph{approximated maximum weighted matching}, \emph{graph colouring}, \emph{maximal independent set}, \emph{approximated maximum independent set}, \emph{approximated minimum dominating set}, and \emph{spanning graphs}~\cite{reingold14new,mansour12converting,even14deterministic, even14best, even15distributed, mansour13local, levi14local,alon12space-efficient,alon10constant,feige15learning, levi15non-constant,rubinfeld11fast,LeviMRRS15,CampagnaGR13}. On the other hand lower bounds in this model are almost nonexistent. In fact, the only lower bound shown directly in the centralised local model is for the spanning graph problem in which the $\lca$ algorithm computes a ``tree-like'' subgraph of a given bounded degree graph~\cite{levi14local,LeviMRRS15}.

\subsection{Overview}

We start this work by introducing the models of computing that we study in Section~\ref{sec:prel}. The key models are the $\local$ model, which is a message-passing model introduced by Linial~\cite{linial92locality}, and the $\lca$ model (centralised local model, a.k.a.\ local computation algorithms), which is a probe-query model introduced by Rubinfeld et al.~\cite{rubinfeld11fast}. In Section~\ref{sec:cont} we briefly list our main contributions. In Section~\ref{sec:simulation} we show our main result, a simulation between the \local{} model and parallel decision trees.
In Section~\ref{sec:lcaparapp} we show that query-order-oblivious \lca\ algorithms are equivalent to  stateless algorithms, and provide a separating example  between stateless algorithms and state-full algorithms with logarithmic space. We also show an application of the main result and provide several new lower bounds in the \lca\ model.
In Section~\ref{sec:constlca} we give an explicit simulation in which a \lca\ algorithm that is allowed to probe anywhere in the graph is simulated by a \lca\ algorithm such that the subgraph induced on its probes is connected. Finally, in Section~\ref{sec:opt-lbs} we give additional lower bounds for the probe complexity of optimisation problems.

\section{Preliminaries} \label{sec:prel}

In this section we describe the various models we discuss in this paper. We focus on {\em problems over labeled graphs} defined as follows. Let $\mathcal{P}$ denote a computational problem
over labeled graphs.  A solution for problem $\mathcal{P}$ over a labeled graph $G$ is a function of which the
domain and range  depend on $\mathcal{P}$ and $G$.  For example, in the maximal
independent set problem, a solution is an indicator function $I\colon V\rightarrow \{0,1\}$ of a
maximal independent set in $G$.
Let $\sol(G,\mathcal{P})$ denote the set of solutions of problem $\mathcal{P}$ over the labeled graph
$G$.

\subsection{\boldmath\local{}: Message-Passing Algorithms}

We use the standard \local{} model~\cite{linial92locality,peleg00distributed} as our starting point. The nodes communicate in synchronous rounds, exchanging messages with all neighbours and doing local computation. Since we are comparing the power of non-local probes, we assume that the size of the graph is known to the nodes, as it usually is in the probe models.

The \local{} model can be defined by the fact that in $t$ communication rounds a node can learn exactly its radius-$t$ neighbourhood. A distributed algorithm with running time $t$ is then a function from the set of possible local neighbourhoods to the set of outputs; see Figure~\ref{fig:models}a.

\begin{figure}
    \centering
    \includegraphics[page=1]{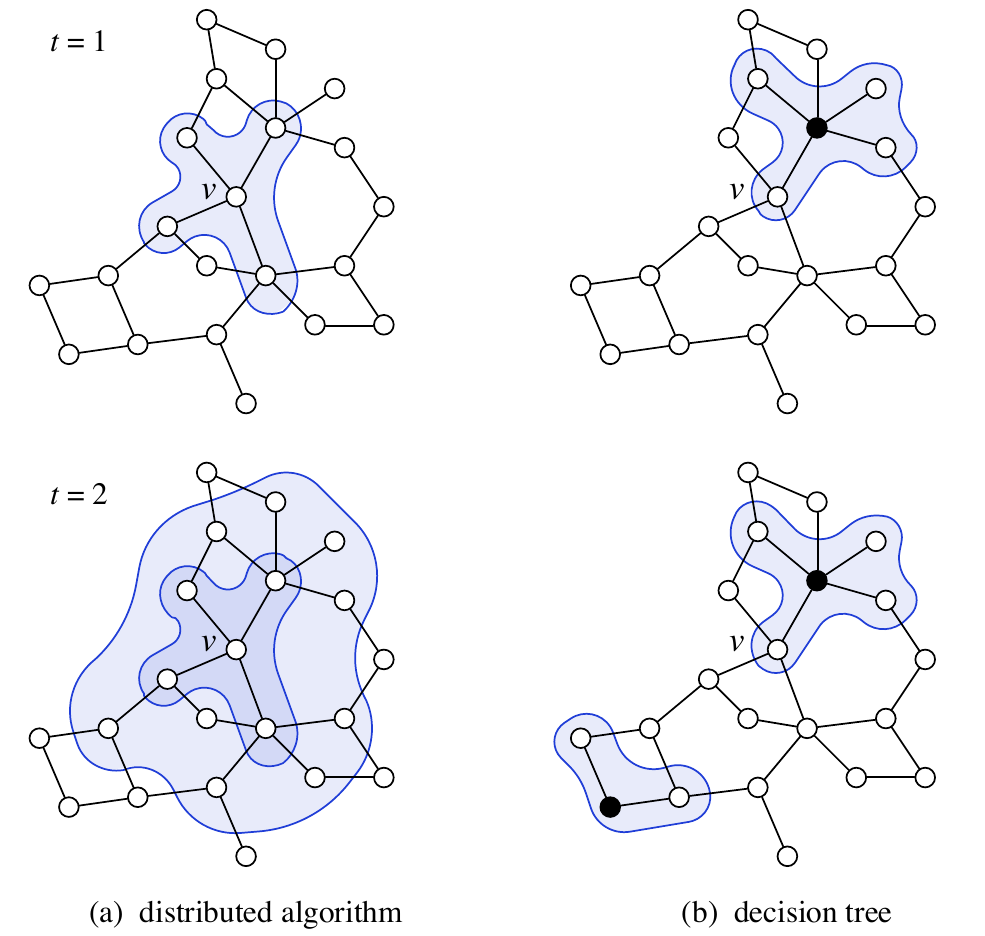}
    \caption{(a)~In a distributed algorithm that runs in time $t$, each node can learn everything up to distance $t$. (b)~In a decision tree algorithm with probe complexity $t$, each node can make up to $t$ probes (indicated with black nodes). After each probe $v$, the decision tree learns the neighbourhood $N(v)$.}\label{fig:models}
\end{figure}

\subsection{\boldmath\partree{}: Parallel Decision Trees}\label{sec:pardec}

We use parallel decision trees, a simple probe model, to connect the \local{} model and the \lca{} model.

Fix an $n$. Our unknown input graph $G = (V,E)$ will have $n$ nodes, labelled with $V = [n] := \{1,2,\dotsc,n\}$. We also refer to the labels as identifiers.
Let $N(v) \subseteq V$ denote the set of neighbours of node $v$. For our purposes, a \emph{decision tree} $\T$ of depth $t$ is an algorithm that can make at most $t$ \emph{probes}: for each probe $v \in V$, tree $\T$ will learn $N(v)$. In words, a decision tree can point at any node of $G$ and ask for a list of its neighbours; see Figure~\ref{fig:models}b. After $t$ probes, tree $\T$ produces an output; we will write $\T(G)$ for the output of decision tree $\T$ when we apply it to graph $G$.

We study graph problems in which the goal is to find a feasible labelling $f\colon V \to L$ of the nodes of our input graph $G$. We say that a graph problem $\mathcal{P}$ can be solved with \emph{$t(n)$ probes in parallel} if there are decision trees $\T_1,\T_2,\dotsc,\T_n$ of depth $t(n)$ such that $f\colon v \mapsto \T_v(G)$ is a feasible solution to problem $\mathcal{P}$ for any input graph $G$. That is, we have $n$ parallel decision trees such that tree $\T_v$ computes the output for node~$v$ and each tree makes at most $t(n)$ probes.

The \emph{probe complexity} of problem $\mathcal{P}$ is the smallest $t(n)$ such that $\mathcal{P}$ can be solved with $t(n)$ probes in parallel.

\subsection{\boldmath\lca{}: Centralised Local Algorithms}

We adopt the definition of the centralised local model (\lca) as formalised in~\cite{rubinfeld11fast}.
In this paper we focus on graph computation problems in which the algorithm is given a probe access to the input graph $G$ through a graph oracle \oracleG.
A \emph{probe} to  \oracleG\ is an identifier of a vertex $v \in V$, in turn, \oracleG\ returns a list of the identifiers of the neighbors of $v$. We assume that for a graph of size $|V| = n$, the set of identifiers is $[n]$.

A \emph{\lca\ algorithm} $\mathcal{A}$ for a computation problem $\mathcal{P}$ is a (possibly randomised) algorithm $\mathcal{A}$ with the following properties.
Algorithm $\mathcal{A}$ is given probe access to the adjacency list oracle \oracleG\ for the input graph $G$, tape of random bits, and local read-write computation memory.
When given an input (query) $x$, algorithm $\mathcal{A}$ returns the answer for $x$.
This answer depends only on $x$, $G$, and the random bits.
The answers given by $\mathcal{A}$ to all possible queries must be consistent; namely, all answers constitute some valid solution to $\mathcal{P}$.

The \emph{probe complexity} is the maximum number of probes that $\mathcal{A}$ makes to $G$ in order to compute an answer for any query.
The \emph{seed length} is the length of the random tape.
The \emph{space complexity} is the local computation memory used by $\mathcal{A}$ over all queries (not including the seed length).
The \emph{success probability} is the probability that $\mathcal{A}$ consistently answers all queries.

We say that a \lca\ algorithm is \emph{stateless} if its space complexity is zero.
We say that a \lca\ algorithm is \emph{state-full} if it is not stateless. With this terminology, we can characterise parallel decision trees as follows:

\begin{fact}\label{fact:lcapar}
  The deterministic stateless \lca\ model is identical to the deterministic parallel decision tree model.
\end{fact}

One can also consider a randomised \partree\ model in which every tree has access to an independent source of randomness. We note that the randomised stateless \lca\ model is stronger than this model since the algorithm has access to the same random seed throughout its entire execution.

\subsection{\boldmath\nice{} Graph Problems}

We say that a problem $\mathcal{P}$ defined over labeled graphs is \nice\ if the following holds:
\begin{enumerate}[noitemsep]
    \item There is a bound on the maximum degree denoted by $\Delta$.
    \item The problem remains invariant under permutation of the labels. Namely, the set of solutions $\sol(G,\mathcal{P})$ is the same for any permutation of the labels of $V(G)$.
    \item For every $f\in \sol(G,\mathcal{P})$ and every  connected component $C$ (maximal connected graph) of $G$, $f$ restricted to $C$ is in $\sol(C,\mathcal{P})$.\label{item:prop3}
\end{enumerate}

We note the family of \nice\ problems includes the so-called \emph{LCL problems} (locally checkable labellings) on bounded-degree graphs \cite{naor95what}.


Examples of \nice\ problems include maximal independent sets, minimal dominating sets, minimal vertex covers, and vertex colouring with $\Delta+1$ colours. With a straightforward generalisation, we can also consider problems in which the goal is to label edges (e.g., maximal matchings and edge colourings)
and problems in which the input graph is labelled (e.g., stable matchings)

\section{Our Contributions} \label{sec:cont}
\paragraph{\boldmath An implicit simulation of parallel decision trees in the \local\ model.}
Our main result is a simulation result connecting the \local{} model to the parallel decision trees.

\begin{restatable}[\partree\ to \local]{theorem}{mainthm} \label{thm:main}
Any \nice\ problem that can be solved in the parallel decision tree model with probe complexity $t(n)$ can be solved in time $t(n^{\log n})$ in the \local{} model provided $t(n) \ll \sqrt{\log n}$.
\end{restatable}

For example, if $t(n)=O(1)$ then the simulation will have the \emph{same} run-time. If $t(n)=\log^* n$ then the simulation will have run-time $\log^* n +1$. Our result also implies a simulation result for a large class of optimization problems.

\paragraph{\boldmath An implicit simulation of \lca\ algorithms in the \local\ model.}
In Section~\ref{sec:lcaparapp} we observe that one can simulate query-order-oblivious \lca\ algorithms by stateless algorithms.
Hence, Theorem~\ref{thm:main} applies to the \lca\ model w.r.t.\ query-order-oblivious deterministic algorithms, formalised as follows.
\begin{restatable}[\lca\ to \local]{theorem}{lblca}
\label{thm:lblca}
  For every query-order-oblivious (or stateless) deterministic \lca\ algorithm $D$ that solves a problem $\mathcal{P} \in $ \nice\, with probe complexity $t(n)=o(\sqrt{\log n})$, there is \local\ algorithm that solves $\mathcal{P}$ by simulating $D$, and for which the number of rounds is at most $t(n)$.
\end{restatable}

As an application of Theorem~\ref{thm:lblca} we show the optimality of several known algorithms as shown in Table~\ref{tab:lb}.

\paragraph{\boldmath A separation between stateless and state-full \lca\ algorithms.}
In Section~\ref{sec:statepower} we show a linear gap in the probe complexity between a \lca\ algorithm with (only) logarithmic state, and a stateless algorithm that computes a leader in variant of the leader election problem.

\paragraph{\boldmath An explicit localised simulation of \lca\ algorithms.}

In Section~\ref{sec:constlca} we show an explicit simulation in which a \lca\ algorithm that is allowed to probe anywhere in the graph is simulated by a \lca\ algorithm that has the following property: the subgraph induced on its probes is connected. Moreover, the overhead in the probe complexity and the seed length is moderate, formalised as follows.
\begin{restatable}[Explicit \lca\ to \lca]{theorem}{mainlocal}
\label{thm:mainlocal}
Let $A$ be a query-order-oblivious $\lca$ algorithm that solves a \nice\ problem $\mathcal{P}$ with probe-complexity $t(n) = O(n^{1/4}/\Delta)$ and seed length $s(n)$.
  Then, there is a query-order-oblivious $\lca$ algorithm $B$ that solves $\mathcal{P}$ by simulating algorithm~$A$. Algorithm $B$ has a probe-complexity of $t(n^4)$ and a seed length of $s(n^4)+O(t(n^4)\cdot \Delta \cdot \log n)$ and the property that it probes within a radius of $t(n^4$); moreover the subgraph induced on the probes of $B$ is connected. The error probability of $B$ equals to the error probability of $A$ plus $O(1/n)$.
\end{restatable}

\paragraph{\boldmath Lower bounds for optimisation problems.}
In addition to the optimization lower bounds implied by Theorem~\ref{thm:main}, in Section~\ref{sec:opt-lbs} we give sublogarithmic ad hoc lower bounds for approximating maximum independent set and maximum cut. This also implies a lower bound on the number of colours required when colouring bipartite graphs. These bounds apply both in the \local{} model and in the parallel decision tree model.


\section{Simulating Probes in the \local{} Model} \label{sec:simulation}

In this section we prove Theorem~\ref{thm:main}, restated here for convenience.
\mainthm*

\subsection{Overview}

Fix an input size $n\in\mathbb{N}$. In order to solve the graph problem on an $n$-node graph $G$ in the \local{} model we simulate the decision tree $\T$ \emph{not} on $G$ directly, but on a much larger graph $G\cup H$ that is the disjoint union of $G$ and some \emph{virtual} graph $H$. The structure of the virtual graph $H$ is agreed upon ahead of time by all the nodes participating in the simulation. Before we invoke the decision tree on $G\cup H$, however, we first \emph{reshuffle} its identifiers randomly. The key idea is to show that this reshuffling trick fools the decision tree: if we simulate $\T$ on a node $v$ and $\T$ probes some node outside of $v$'s local neighbourhood, then the probe lands in $H$ with overwhelming probability. Since all nodes in $G$ know the structure of $H$, they can answer such ``global'' probes consistently. Finally, we argue that there is some fixed choice of randomness that makes the simulation succeed on all instances of order $n$. This makes the simulation deterministic.

\subsection{Setup}

Let $N = n^{\log n}$. The running time of our simulation will be $t(N)\ll \log n$.

Let $H$ be any bounded-degree graph on the vertex set $[N]\smallsetminus[n]$. For example, if we are studying a colouring problem on cycles, we can let $H$ be a cycle. This graph remains fixed and does not depend on the $n$-node input graph $G$. In this sense, all the nodes in the simulation will know the structure of $H$.

Let $\pi\colon[N]\to[N]$ be a random permutation. Again, all the nodes in the simulation are assumed to agree on the same $\pi$. (The simulation can be thought of as being defined with respect to the outcomes of some public random coins.)

\subsection{Simulation}\label{sec:simupar}

During the simulation we will always present the decision tree with a relabelled version of the graph $G\cup H$ where a node $v$ has identifier $\pi(v)$. That is, when $\T$ probes an identifier $w$ this is interpreted as $\pi^{-1}(w)$ and when we respond with node $v$ we relabel its identifier as $\pi(v)$. Note that each node in $G$, knowing $\pi$, can perform this translation locally.

A node $v\in V(G)$ starts its simulation by invoking the decision tree on $\pi(v)$ with the hope of answering its $t(N)$ probes based on the $t(N)$-neighbourhood of $v$ in $G$ and knowledge of $H$. During this attempt we maintain a set of \emph{discovered nodes} $Q$ that contains all the nodes whose relabelled versions have been sent back to $\T$ in response to a probe. Note that $|Q| \leq k$ where $k \triangleq 1 + (\Delta+1) \cdot t(N)$ is an upper bound on the total number of inspections of $\pi$ performed by the simulation. Initially $Q=\{v\}$, so that only $v$ is discovered. We will ensure that the following invariant holds throughout the simulation:
\begin{quote}
\emph{Simulation invariant:}
After the $s$-th step every member of $Q$ is either in $H$ or at most at distance $s$ from $v$.
\end{quote}
Suppose $\T$ probes a relabelled identifier $w$ corresponding to a node $u = \pi^{-1}(w)$. We have two cases:
\begin{itemize}[itemsep=5pt,label=$\bullet$]
\item {\bf Local probe, \boldmath$u\in Q$.} In this case we are always successful: by the invariant, the neighbours of $u$ are known to $v$ so we can add them to $Q$ and return their relabelled versions to $\T$.
\item {\bf Global probe, \boldmath$u\notin Q$.} Two sub-cases depending on whether $u\in V(H)$:
\begin{itemize}[topsep=5pt,noitemsep,label=$-$]
\item {\boldmath $u\in V(H)$:} Success! The structure of $H$ is known to $v$ so we can add the neighbours of $u$ to $Q$ and return their relabelled versions to $\T$.
\item {\boldmath $u\notin V(H)$:} Here we simply say that the simulation has \emph{failed} and we terminate the simulation. (Note that by this convention we may fail even if $u$ is in the $t(N)$-neighbourhood of $v$; however this convention helps us maintain the invariant.)
\end{itemize}
\end{itemize}
When the decision tree finally returns an output label for $\pi(v)$, we simply output the same label for $v$.

\subsection{Failure Probability}
Next, we analyse the probability that our simulation fails when invoked on a particular node. Suppose we need to respond to a global probe sometime in our simulation. That is, the probe is to a node whose relabelled identifier is $w$ and the node has not been discovered yet, i.e., $w\notin \pi(Q)$. By the principle of deferred decisions, we can think of the node $u=\pi^{-1}(w)$ as being uniformly distributed among the undiscovered nodes $[N]\smallsetminus Q$. Hence the failure probability (conditioned on any outcome of $\pi(Q)$) is
\begin{equation}
\Pr[u\in V(G)] = \frac{|V(G)\smallsetminus Q|}{|[N] \smallsetminus Q|} \leq \frac{n}{N-k}
\end{equation}
Consequently, by the union bound \begin{equation}\label{eq:simfails}
\Pr(\text{simulation fails at some probe step}) \leq k \cdot \frac{n}{N-k}\ ,
\end{equation}
which is at most $\frac{n^2}{N}$ for any $k \leq n/2$.
\subsection{Derandomisation}

Next, we would like to argue that there is a fixed choice for~$\pi$ that makes the simulation succeed on all $n$-node graphs $G$. To this end, we note that on any run of $\T$, the final output of $\T$ depends on at most $k$ nodes (and their adjacency relations) in the input graph. There are at most $n^{\Delta k}$ graphs on $k$ nodes having identifiers from $[n]$; let $\G$ consist of all such graphs. For technical convenience we can add dummy nodes to each graph in $\G$ to make their vertex set be~$[n]$. Executing our simulation on each node ($n$ choices) of each graph in $\G$ ($|\G|$ choices) the probability that some simulation fails (each fails with probability $\leq n^2/N$) is at most
\[
n\cdot|\G|\cdot n^2 / N = n^{-\Omega(\log n)} \ll 1
\]
by yet another union bound. Thus, we can find a fixed outcome of $\pi$ for which the simulation succeeds simultaneously on all of $\G$ and therefore on all $n$-node graphs.

\subsection{Correctness}

It remains to point out that the output labelling $f$ produced by the simulation constitutes a feasible solution to the graph problem under consideration. Here it suffices to assume that the graph problem satisfies the following property: if $f$ is a feasible solution for a graph $G$ and $C\subseteq G$ is a connected component of $G$, then the restriction of $f$ to $V(C)$ is a feasible solution for $C$ and that this remains so under any relabelling of the identifiers. In particular, all \nice\ problems satisfy this property.

\subsection{Simulation for Optimization}\label{sec:optim}

We note that for the correctness we used a weaker property than property~(\ref{item:prop3}) of \nice\ problems. In fact the correctness applies to any problem $\mathcal{P}$ such that for all $n$ there exists a graph $H$ on $n$ vertices and with maximum degree at most $\Delta$ such that for every $f \in \sol(H\cup G, \mathcal{P})$ we have that $f$ restricted to $G$ is in $\sol(G, \mathcal{P})$.
By taking $H$ to be the graph with no edges we obtain that Theorem~\ref{thm:main} can be generalised to, for example, $(1-\eps)$-approximated maximum (weighted) matching and approximation of vertex covers.


\section{Centralised Local Model and Parallel Decision Trees}\label{sec:lcaparapp}

In this section we observe that one can simulate query-order-oblivious \lca\ algorithms by stateless algorithms.
Recall that a deterministic stateless \lca\ algorithm is a \partree\ algorithm (Fact~\ref{fact:lcapar}). Hence, the simulation result (Theorem~\ref{thm:main}) applies to the \lca\ model w.r.t.\ query-order-oblivious deterministic algorithms.
We then summarise the obtained lower bounds and show the optimality of several known algorithms.
We conclude with a separation between stateless and state-full \lca\ algorithms: we show a linear gap in the probe complexity between a \lca\ with (only) logarithmic state and a stateless algorithm that computes a leader in variant of the leader election problem.

\subsection{\boldmath Query-Order-Oblivious vs. Stateless \lca\ Algorithms}

We say that a \lca\ algorithm is \emph{query-order-oblivious} if the (global) solution that the algorithm computes does not depend on the input sequence of queries.
Even et al.~\cite{even14best} state that a stateless \lca\ algorithm is also query-order-oblivious.
In the following lemma we prove that the converse is also true, that is, for every query-order-oblivious \lca\ algorithm there is a stateless algorithm that can simulate it.

\begin{lemma}\label{lemma:qoo}
  For every query-order-oblivious \lca\ algorithm $C$ there is a stateless \lca\ algorithm $S$ that simulates $C$. Moreover, the probe complexities of $C$ and $S$ are equal.
\end{lemma}
\begin{proof}
  The stateless algorithm $S$ simply invokes $C$ with its initial state and the same random seed for every input query.
\end{proof}

\subsection{\boldmath \lca\ vs. \local}\label{sec:lcaapp}
Parnas and Ron~\cite{parnas07approximating} observed that given a deterministic \local\ algorithm $D$ that performs $r$ rounds of communication, there is a \lca\ algorithm $C$ that simulates $D$ with probe complexity which is $O(\Delta^r)$.
Even et al.~\cite{even14best} observed that if a deterministic \lca\ algorithm $C$ probes in an $r$-neighborhood of each queried vertex, then there is a deterministic \local\ algorithm $D$ that simulates $C$ in $r$ communication rounds.
Theorem~\ref{thm:main} implies that for some \lca\ algorithms there is a \local\ (implicit) simulation such that the number of communication rounds is asymptotically equal to the probe complexity, even though the \lca\ algorithm probes outside of the $r$-neighborhood of a queried vertex. This argument allows carrying lower bounds to the \lca\ model from the \local\ model. In Section~\ref{sec:constlca} we show an explicit simulation in which a \lca\ algorithm that is allowed to probe anywhere in the graph is simulated by a \lca\ algorithm that probes along connected components only.

Theorem~\ref{thm:main} and Fact~\ref{fact:lcapar} along with Lemma~\ref{lemma:qoo} imply the following theorem.
\lblca*

In Section~\ref{sec:optim} we observed that Theorem~\ref{thm:main} also applies to several optimization problems.
In Table~\ref{tab:lb} we summarise (1)~the known \lca\ algorithms, their probe complexities as well as the obtained approximation ratios, and (2)~corresponding \local\ lower bounds. By Theorem~\ref{thm:lblca} all stated lower bounds apply to deterministic, query-order-oblivious \lca\ algorithms.

\begin{table}[htb]
\centering
\begin{tabular}{lllll}
\toprule
Problem
 & \multicolumn{2}{l}{\lca\ upper bounds}
 & \multicolumn{2}{l}{\local\ lower bounds} \\
 &  (deterministic, 0-space) \\
 & \# probes && \# rounds \\
\midrule
\mis     & $O(\log^* n)$ & \cite{even14best} &  $\Omega(\log^* n)$ & \cite{linial92locality}\\
\mm     & $O(\log^* n)$ & \cite{even14best} &  $\Omega(\log^* n)$ & \cite{linial92locality}\\
\dcolor  & $O(\log^* n)$ & \cite{even14best} &  $\Omega(\log^* n)$ & \cite{linial92locality}\\
$(1-\eps)$-\mcm & $\poly(\log^* n)$ & \cite{even14best} & $\Omega(\log^* n)$ & \cite{lenzen2008leveraging,czygrinow2008fast}\\
$(1-\eps)$-\mwm & $\poly(\min\{\Gamma,n/\eps\}\cdot\log^*n)$ & \cite{even14best} & $\Omega(\log^* n)$ & \cite{lenzen2008leveraging,czygrinow2008fast}\\
\bottomrule
\end{tabular}
\caption{\mis\ denotes maximal independent set, \mm\ denotes maximal matching, \dcolor\ denotes $\Delta+1$ vertex colouring,  $(1-\eps)$-\mcm\ denotes $(1-\eps)$-approximated maximum cardinality matching, and
  $(1-\eps)$-\mwm\ denotes $(1-\eps)$-approximated maximum weighted matching.
All the upper bounds presented in this table are of algorithms which are deterministic and stateless.
All the upper bounds are presented under
the assumption that $\Delta = O(1)$ and $\eps = O(1)$.
For weighted graphs, the ratio between the maximum to
minimum edge weight is denoted by $\Gamma$.
}\label{tab:lb}
\end{table}

\subsection{\boldmath The Power of State-full \lca\ Algorithms}\label{sec:statepower}
In this section we prove that the ratio between the probe complexity of a stateless \lca\ algorithm and a state-full algorithm is $\Omega(n)$  in a variant of the leader-election problem, defined as follows. This linear gap occurs even if the state-full algorithm uses only logarithmic number of bits to encode its state.

\begin{definition}[Two-path leader-election]
The input is a graph $G=(V,E)$, where $|V| \geq 6$.
In this graph there are two vertex-disjoint paths $p_i=(v^i_1,v^i_2,v^i_3)$ for $i \in \{1,2\}$. The rest of the vertices are of degree zero.
The ID set is $\{1,\ldots,n\}$ where $n$ is the number of vertices.\footnote{Note that the Two-path leader-election problem requires $\Omega(n)$ probes in the stateless \lca\ model even when the ID domain is exactly $\{1, \ldots, n\}$.}

The local output is $1$ in exactly one vertex in $\{v^1_2, v^2_2\}$. All the the other vertices output $0$.
\end{definition}

Note that a state-full \lca\ algorithm requires at most $O(\log n)$ space. The state-full algorithm proceeds as follows. The algorithm output ``no'' if the queried vertex has degree $0$ or~$1$. Otherwise, if the queried vertex is the first ``middle'' vertex, then it outputs $1$ and writes its ID to the state.
The next queried vertices will output $0$, even if their degree is $2$.

The linear gap between stateless \lca\ algorithms and state-full algorithms is formalised in the following lemma.
\begin{theorem}
  Every deterministic stateless \lca\ algorithm requires $\Omega(n)$ probes to compute a leader in the Two-path leader-election problem.
\end{theorem}
\begin{proof}
Let $\mathcal{A}$ be a deterministic algorithm for the Two-path leader-election problem.
Let the triplet $\ell^i = (\ell_1^i, \ell_2^i, \ell_3^i)$ denote the labeling of $p_i$. Namely, $\ell_j^i$ is the identifier of $v_j^i$.
Let $m$ denote the size of the domain of $\ell^i$, namely the number possible labelings of $p_i$, i.e., $m= n (n-1) (n-2) = \Theta(n^3)$.
Observe that the answers to probes on $G$ depend only on the labels of $p_1$ and $p_2$ and the same is true for the output of $\mathcal{A}$ (every vertex that is not in $p_1$ or $p_2$ is an isolated vertex).

The proof has two parts. In the first part we prove that there is a labeling of $p_1$ and a constant $c$ such that $\mathcal{A}$ outputs $1$ on $v_2^1$ for at least $cN$ of the possible labelings of $p_2$ and $0$ on $v_2^1$ for at least $cN$ of the possible labelings of $p_2$.
In the second part we assume in contradiction that the algorithm requires less than $\Omega(n)$ probes and show that there exists a legal labeling of $p_1$ and $p_2$ such that $\mathcal{A}$ errs.

We now give the first part of the proof.
Let $T_1, \ldots, T_N$ denote all possible labelings of $p_1$.
Consider the $N$ by $N$ matrix $M$ that encodes the decision of the deterministic algorithm, that is, $M(i,j)$ equals to the algorithm's output for the query $v_2^1$ when $p_1$ is labeled by $T_i$ and $p_2$ is labeled by $T_j$.
Note that some pairs of labelings are illegal (their intersection should be empty). Let $N' = \Theta(n^3) = \gamma n^3$ denote the number of legal labels in each row (or column) of $M$.
Let $S$ be the set of rows in which the number of $1$'s is in $[\frac 13 N', \frac 23 N']$. If $S \neq \emptyset$ then we are done.
Otherwise, let $L_0$ ($L_1$) be the set of rows in which the number of $0$'s ($1$'s) is greater than the number of $1$'s ($0$'s). Before we continue with the proof, we prove the following property of $L_0$ and $L_1$.

\begin{lemma}
If $S = \emptyset$ then $L_i \geq N/4$ for every $i\in \{0, 1\}$.
\end{lemma}
\begin{proof}
From the fact that the algorithm always outputs $1$ on exactly one of the elements in $(v_2^1, v_2^2)$ we derive that in $M$ the number of $1$'s equals to the number of $0$'s.
Now, assume for the sake of contradiction that $L_i < N/4$ for some $i \in \{0,1\}$.
Since every row in $L_i$ has at most $N'$ $i$'s we obtain that the total number of $i$'s in $M$ is at most $|L_i| \cdot N' + (N- |L_i|) \cdot (N'/3) < (N/2) N'$, in contradiction to the fact that the number of $i$'s is exactly $(N/2) N'$.
\end{proof}

Let $i\in \{0,1\}$. Consider $M$ restricted to $L_i$.
Since in each row in $L_i$ the number of $i$'s is at least $\frac 23 N'$ we obtain that the total number of $i$'s in this restriction is at least $\frac 23 N' |L_i|$.\
Let $X$ denote average of the number of $i$'s in each column restricted to $L_i$.
Thus we obtain that $X \geq \frac 23 N' |L_i|\frac{1}{N}$.
Let $c \triangleq \frac{2|L_i|}{N'}(\frac23 \frac{N'}{N} -\frac 12)$, and let $a$ be the fraction of columns for which the number of $i$'s in each column restricted to $L_i$ is greater than $cN'$.
Then $X \leq |L_i|\cdot a + cN' (1-a)$.
Assume towards contradiction that $a \leq 1/2$.
Hence, $X \leq |L_i|/2 + cN'/2 < 2/3 N' |L_i|\frac{1}{N}$, a contradiction.
Hence, there exists a constant $c$ such that for strictly more than $1/2$ of the columns it holds that at least $cN'$ elements are $1$ (0).
We derive that there exists a columns such that both the number of $1$'s and the number of $0$'s is at least $cN'$ as desired.

Consider the query $v^1_2$. Observe that in each probe for which the answer is ``isolated vertex'', the algorithm eliminates at most $3\cdot n^2$ possible labeling of $p_2$.
Thus, after at most $c\gamma n^3/ (3n^2) = \Omega(n)$ probes, at least one ``zero'' labeling and one ``one'' labeling remains, which concludes the proof of the lemma.
\end{proof}

\section{\boldmath Localizing Stateless \lca\ Algorithms} \label{sec:constlca}
In this section we give a constructive, polynomial blow-up, randomised \lca\ simulation.

\mainlocal*

\subsection{Novelty}
We observe that in the \lca\ model an explicit simulation is possible, since, unlike the \local\ and \partree\ models, a \lca\ algorithm uses the same random seed for all queries.
This random seed is a costly resource and we try to minimize its length.
Known randomised implementations of greedy algorithms in the \lca\ domain require explicit random ordering constructions~\cite{alon12space-efficient,reingold14new} over the vertices or edges~\cite{alon12space-efficient,mansour12converting,mansour13local,reingold14new}.
In our implementation of the simulation, on the other hand, we use a permutation over the labels, which is a stronger requirement than a random ordering.
This requirement comes from the fact that in the simulation each vertex has a unique identifier and that the set of identifiers is assumed to be known.
In what follows we propose an implementation which builds on techniques by~\cite{v009a015,knr09}.

For many graph problems the simulation presented in this section proves that one can design a \lca\ algorithm that only performs ``close'' probes. On one hand, this brings closer the two models of  \lca\ and \local. On the other hand, the single source of randomness that a \lca\ algorithm possesses is an advantage for the \lca\ algorithm, so they are still far apart.

\subsection{Notations and Definitions}
Let $S_n$ denote the set of all permutations on $[n]$.
\begin{definition}[Statistical Distance]
Let $D_1$, $D_2$ be distributions over a finite set $\Omega$. The statistical distance between $D_1$ and $D_2$ is
$$
\|D_1 - D_2 \|= \frac 12 \sum_{\omega\in \Omega} |D_1(\omega) - D_2(\omega)| .
$$
Alternatively,
$$
\|D_1 - D_2 \|= \max_{A\subseteq \Omega} \left|\sum_{\omega\in A} D_1(\omega) - \sum_{\omega \in A} D_2(\omega)\right| .
$$
\end{definition}
We say that $D_1$ and $D_2$ are $\epsilon$-close if $\|D_1 - D_2\| \leq \epsilon$.

\begin{definition}
  Let $n,k \in \mathbb{N}$, and let $\mathcal{F} \subseteq S_n$ be a multiset of permutations. Let $\eps \geq 0$. The multiset $\mathcal{F}$ is $k$-wise $\eps$-dependent if for every $k$-tuple of distinct elements $(x_1,\ldots,x_k) \in [n]^k$, the distribution $(f(x_1),
  f(x_2),\ldots,f(x_k))$, when $f \simur \mathcal{F}$ is $\eps$-close to the uniform distribution over all $k$-tuples of distinct elements of $[n]$.
\end{definition}
As a special case, a multiset of permutations is $k$-wise independent if it is $k$-wise $0$-dependent.

\subsection{Techniques}
We shall use the following results from previous work in the proof of Theorem~\ref{thm:mainlocal}.
The following theorem is an immediate corollary of a theorem due to Alon and Lovett~\cite[Theorem~1.1]{v009a015}.
\begin{theorem}\label{thm:alon}
Let $\mu$ be a distribution taking values in $S_n$ which is $k$-wise $\epsilon$-dependent. Then there exists a distribution $\mu'$ over permutations which is $k$-wise independent, and such that the
statistical distance between $\mu$ and $\mu'$
is at most $O(\epsilon n^{4k})$.
\end{theorem}

Additionally, we build on the following construction in our \lca\ simulation \simu. This construction enables accesses to a  uniform random permutation from a $k$-wise $\epsilon$-dependent family of permutations by using a seed of length $O(k\cdot \log n  + \log(1/\epsilon))$.
\begin{theorem}[{\cite[Theorem~5.9]{knr09}}]\label{thm:kaplan}
There exists $\mathcal{F} \subseteq S_n$, such that $\mathcal{F}$ is $k$-wise $\epsilon$-dependent. $\mathcal{F}$ has description length $O(k\cdot \log n  + \log(1/\epsilon))$, and time complexity $\poly(\log n, k, \log (1/\epsilon))$.
\end{theorem}

\subsection{\boldmath Simulating Queries in the \lca\ Model}
Recall that the simulation in Section~\ref{sec:simupar} simulates a \partree\ algorithm via a distributed \local\ algorithm.
In order to prove Theorem~\ref{thm:mainlocal} we consider a similar simulation with the difference that now both the simulation and the simulated algorithm are in \lca.
However, the simulation has the property that it is limited to query \oracleG\ only on $Q$ (as defined in Section~\ref{sec:simupar}).
Since a \lca\ is equipped with a random seed we can use a randomised simulation and consequently the blow-up in the probe-complexity will be significantly smaller.
In particular, in this section we consider $N = O(n^4)$, that is, the size of the augmented graph $G \cup H$ is polynomial in the size of the input graph $G$.
We are interested in keeping the seed length small and so we show next that the additional random seed that is required for the simulation is small.

\subsection{Proof of Theorem~\ref{thm:mainlocal}}
Let $A$ be as in Theorem~\ref{thm:mainlocal} and let \simuA\ denote its simulation on $G\cup H$ as described  above.
Since the size of $G\cup H$ is $N$ the size of the random seed required by $A$ is $s(N)$.
For a fixed a fixed random seed $r\in \{0,1\}^{s(N)}$,
a fixed query $q\in V(G)$ and a fixed permutation $\pi\in S_N$ let $\simu(G, r, \pi, q)$ be the indicator variable for the event that the simulation succeeds in simulating $A$ with random seed $r$ on input $q$ where $\pi$ is the permutation which is used to relabel the vertices in $G\cup H$.

Let $\mathcal{F}$ be a family of $k$-wise independent permutations over $[N]$ where  $k \triangleq 1 + (\Delta+1) \cdot t(N)$.

\begin{lemma}\label{lemma:unitokwise}
For every $q \in V(G)$ and $r\in \{0,1\}^{s(N)}$,
$$\Pr_{\pi\simur \mathcal{F}}(\simu(G, r, \pi, q) = 0) \leq \frac{kn}{N-k} $$
\end{lemma}
\begin{proof}
Fix $r\in \{0,1\}^{s(N)}$, $\pi\in S_N$ and $q\in V(G)$.
We will prove that
$$
\Pr_{\pi\simur \mathcal{F}}(\simu(G, r, \pi, q) = 0) = \Pr_{\pi\simur S_N}(\simu(G, r, \pi, q) = 0),
$$
and then the lemma will follow from Equation~\ref{eq:simfails}.
Recall that the simulation relabels the vertices by accessing both to $\pi$ and $\pi^{-1}$.
Now consider the sequence of inspections the simulation makes to $\pi$ and $\pi^{-1}$ by the order they occurred.\footnote{The construction stated in Theorem~\ref{thm:kaplan} computes $\pi(v)$ for some $v$ in time $\poly(\log n, k, \log (1/\epsilon))$. This construction does not describe a time efficient way to access $\pi^{-1}$. As time complexity is not the focus of this paper, we implement the inverse access in a straightforward manner.}
The first inspection to $\pi$ is $\pi(q)$.
Then, according to the decisions of $A$ and the answers from \oracleG, the simulation continues to inspect both $\pi$ and $\pi^{-1}$ on at most $k-1$ locations.

Let $(v, u, b) \in N \times N \times \{-1, 1\}$ represent a single inspection and answer as follows:
If $b = 1$ then the interpretation is that the simulation inspects $\pi$ at index $v$ and the answer is $u$ and if $b = -1$ then the interpretation is that the simulation inspects $\pi^{-1}$ at $v$ and the answer is $u$.

For a fixed $G$, $r$, $\pi$ and $q$ the sequence of (possibly adaptive) inspections and answers is fixed. Let it be denoted by $\sigma = (v_1, u_1, b_1), \ldots,(v_{\ell}, u_{\ell}, b_{\ell})$ where $\ell \leq k$.
We say that a permutation $\pi$ agrees with the sequence $\sigma$ iff $\pi^{b_i}(v_i) = u_i$ for every $i\in [\ell]$.
Clearly, we can replace $\pi$ with any permutation $\pi'$ which agrees with $\sigma$ and the outcome (failure or success) remains unchanged, that is $\simu(G, r, \pi, q) = \simu(G, r, \pi', q)$.
We say that $\sigma$ is {\em positive} if $b_i = 1$ for every $i\in [\ell]$.
We say that a pair of sequences $\sigma$ and $\sigma'$ are {\em equal} if any permutation that agrees with $\sigma$ also agrees with $\sigma'$ and vice versa.
Observe that from any sequence $\sigma$ we can obtain an equal sequence, $\sigma'$, which is positive by performing the following replacements: If there exists $j$ such that $b_j = -1$ (this means that $\pi^{-1}(v_j) = u_j$ which is equivalent to $v_j = \pi(u_j)$) then replace $(v_j, u_j, -1)$ with $(u_j, v_j, 1)$.
Therefore we obtain that $S_N$ can be partitioned into equivalence classes where in each class $C$:
(1)~all the permutations agree with some positive sequence $Q$ of length at most $k$, and (2)~$\simu(G, r, \pi, q)$ is fixed when $\pi$ is taken from $C$.
Since for every positive sequence $\sigma$ we have
$$
\Pr_{\pi\simur \mathcal{F}}[\pi\text{ agrees with } \sigma] = \Pr_{\pi\simur S_N}[\pi\text{ agrees with } \sigma]\ ,$$
we obtain the desired result.
\end{proof}

\begin{corollary}\label{coro:n3}
For a fixed graph $G$ and a fixed random seed $r \in \{0,1\}^{s(N)}$,
the probability that the simulation works when $\pi \simur \mathcal{F}$ is at least $1- O(1/n)$.
\end{corollary}

\begin{proof}
Recall that the simulation succeeds if and only if $\simu(G, r, \pi, q) = 1$ for every $q \in V(G)$.
Since $t(N) = O(n/\Delta)$, from Lemma~\ref{lemma:unitokwise} and the union bound over the $n$ possible queries we obtain that the simulation succeeds with probability at least $1 - \frac{kn^2}{N-k} = 1 - O(1/n)$.
\end{proof}

Let $\mathcal{F}'$ be a $k$-wise $\epsilon$-dependent family of permutations.
\begin{corollary}\label{coro:simerr}
For a fixed graph $G$ and a fixed random seed $r \in \{0,1\}^{s(N)}$,
the probability that the simulation works when $\pi \simur \mathcal{F}'$ is at least $1- O(1/n) - O(\epsilon N^{4k})$.
\end{corollary}
\begin{proof}
For a fixed graph $G$ and random seed $r$ let $S\subseteq S_N$ denote the subset of permutations $\pi$ for which the simulation works.
By Corollary~\ref{coro:n3}, $\Pr(\pi \in S) \geq 1- O(1/n)$ when $\pi \simur \mathcal{F}$.
By Theorem~\ref{thm:alon} the statistical distance between the uniform distribution over $\mathcal{F}$ and the uniform distribution over $\mathcal{F}'$ is at most $O(\epsilon N^{4k})$.
Therefore $\Pr(\pi \in S) \geq 1- O(1/N)-O(\epsilon N^{4k})$ when $\pi \simur \mathcal{F}'$ as desired.
\end{proof}

We are now ready to prove the main result of this section.
\begin{proof}[Proof of Theorem~\ref{thm:mainlocal}]
In the simulation we shall use the construction from Theorem~\ref{thm:kaplan} to obtain a random access to a permutation $\pi$ over $[N]$ where $N = n^4$, using a random seed of length $O(k \cdot \log N + \log \gamma)$ where $\gamma = O(n N^{4k})$ and access time $\poly(\log N, k, \log \gamma)$.
By Corollary~\ref{coro:simerr} the simulation works with probability at least $1 - O(1/n)$ as desired.
\end{proof}

\section{Approximability with Parallel Decision Trees} \label{sec:opt-lbs}

In this section we prove almost tight bounds for the approximability of maximum independent sets and maximum cuts in the decision tree model. The first result also implies a lower bound on the number of colours used by any sublogarithmic-time algorithm.

The result for independent set is essentially due to Alon~\cite{alon10constant}, who uses the fact that there are two graph distributions with different properties that are indistinguishable to a constant-time algorithm. We strengthen this technique by showing that it still holds even if the algorithm is given information on the global structure of the graph. All results hold also for the \local{} model with suitable adjustments to the constant factors. There is also a simple extension to randomised algorithms for all results.

\begin{theorem} \label{thm:adh-mis-apx}
    Maximum independent set cannot be approximated within a factor of $\frac{4 \ln d}{d}+\varepsilon$, for any $\varepsilon > 0$, in $d$-regular graphs with a sublogarithmic number of probes, even if the instance is known to be isomorphic to one of two possible instances.
\end{theorem}

This is almost tight in the sense that the simple randomised greedy algorithm finds an independent set of expected size $\Omega(\ln d / d)$ in triangle-free graphs of average degree $d$~\cite{shearer83note}. This centralised algorithm can be simulated in constant time in a distributed fashion for example using the random ordering technique of Nguyen and Onak so that an expected $(1-\varepsilon)$-fraction of the nodes succeeds in the simulation~\cite{nguyen08constant-time}.

The proof of Theorem~\ref{thm:adh-mis-apx} will immediately yield the following corollary.
\begin{corollary} \label{cor:col-lb}
    Bipartite graphs with nodes of degree $d$, including trees, cannot be coloured with $o(d/\ln d)$ colours with a sublogarithmic number of probes.
\end{corollary}

Finally, we will show the following, almost optimal bound for maximum cut.
\begin{theorem} \label{thm:max-cut-apx}
    Maximum cut cannot be approximated within a factor of $\frac{1}{2}+\frac{\sqrt{d-1}}{d}$ in $d$-regular graphs, for infinitely many values of $d$, with a sublogarithmic number of probes.
\end{theorem}

The lower bound is almost tight in the sense that there are simple randomised distributed constant-time algorithms that cut a fraction of $1/2 + \Omega(1/\sqrt{d})$ edges in triangle-free graphs~\cite{hirvonen14local-maxcut,shearer92bipartite}.

All three proofs rely on the fact that there exist $d$-regular graphs of a logarithmic girth with no good solutions. For independent set (and chromatic number) these can be shown to exist using arguments about random graphs; for maximum cut, we rely on the explicit constructions of Lubotzky et al.~\cite{lubotzky88ramanujan} and Morgenstern~\cite{morgenstern94existence}. Formally, we use the following theorem.
\begin{theorem}[\cite{bollobas81independence,alon10constant,lubotzky88ramanujan,morgenstern94existence,mohar90eigenvalue}] \label{thm:lb-graphs} The following infinite families of graphs exist:
\begin{enumerate}[label=(\roman*)]
    \item For each $d \ge 3$, there is a family of graphs $\G_d = (G_i)_{i \ge 1}$ such that each $G_i$ has a girth of $g(G_i) \ge 0.5 \log_{d-1} |G_i|$ and the size of the largest independent set is $\alpha(G_i) \le \frac{2 \ln d}{d}|G_i| - O(\sqrt{|G_i|}) $.\label{thm:lb-graphs-is}
    \item For each prime power $d-1 \ge 2$, there is a family of graphs $\H_d = (H_i)_{i \ge 1}$ such that each $H_i$ has a girth of $g(H_i) = \Omega(\log_{d-1} |H_i|)$ and the fraction of cut edges in the largest cut is $\mc(H_i) \le \frac{1}{2}+\frac{\sqrt{d-1}}{d}$.\label{thm:lb-graphs-mc}
\end{enumerate}
\end{theorem}

\begin{proof}
    For a proof of part~\ref{thm:lb-graphs-is}, see \cite{bollobas81independence}~and~\cite[Theorem 1.2]{alon10constant}. For part~\ref{thm:lb-graphs-mc}, Ramanujan graph constructions~\cite{lubotzky88ramanujan,morgenstern94existence} have no large cuts~\cite{mohar90eigenvalue}.
\end{proof}

We will use the same strategy in the proofs of Theorems~\ref{thm:adh-mis-apx}~and~\ref{thm:max-cut-apx}: we show that there exist two graphs, one with a small optimal solution and another with a large optimal solution, such that a $t$-probe decision tree cannot separate the two.

\begin{proof}[Proof of Theorem~\ref{thm:adh-mis-apx}]
    Let $\G_d$ be a family of graphs as in part~\ref{thm:lb-graphs-is} of Theorem~\ref{thm:lb-graphs}. For an arbitrary $G \in \G_d$ with $|G| = n$, we have the following construction. Let $A_G$ be the graph consisting of two disjoint copies of $G$, that is, for each node $v \in G$ there are two copies $v_1, v_2 \in V(A_G)$ and for each edge $\{u,v\} \in E(G)$ two copies $\{u_1, v_1\}, \{u_2, v_2\} \in E(A_G)$. Similarly, let $B_G$ be the \emph{bipartite double cover} of $G$, that is, a graph on $2n$ nodes such that for each node $v \in V(G)$ there are two copies $v_1, v_2 \in V(B_G)$ and for each edge $\{u,v\} \in E(G)$ there are two edges $\{u_1, v_2 \}, \{u_2, v_1 \} \in E(H)$. We identify $V = V(A_G) = V(B_G)$ in the natural way. Note that clearly $\alpha(B_G) = n$.

    Let $\id \colon V \to [2n]$ be an identifier assignment on $A_G$ and $B_G$. We let the algorithm know $G$, $A_G$, $B_G$ and $\id$. The input is then promised to be either $A_G$ or $B_G$ with the following random perturbation $p$: for each pair of nodes $v_1$ and $v_2$ independently, we swap the identifiers of the nodes with probability $1/2$. Denote the graph distributions produced this way by $p(A_G)$ and $p(B_G)$.

    Denote by $Q(v,t)$ the sequence of $t(n)$ probes done by $\T_v$, where $t(n) = o(\log n)$. This sequence can be seen as a subset of edges revealed to the algorithm. Now it is easy to see for an arbitrary node $v$ that, since the graph induced by nodes in $Q(v,t)$ does not contain any cycles (that is, $t < g(G)/2$), we have $\Pr[Q(v,t) | p(A_G)] = \Pr[Q(v,t) | p(B_G)]$; consider a rooted version of each subtree in $Q(v,t)$ and note that both $p(A_G)$ and $p(B_G)$ have the same probability to produce that subtree.

    Finally, given that the distribution of the outputs of $\T_v$ is the same in $p(A_G)$ and $p(B_G)$ for each node $v$, and that any $t(n) = o(\log n)$ time algorithm $A$ must have $\sum_{v \in V(A_G)} \E[\A(v)] = \sum_{v \in V(B_G)} \E[\A(v)] \le 2\alpha(G)$, there must exist an identifier assignment on $B_G$ such that the size of the independent set produced by $\A$ is at most the expectation, that is $2\alpha(G)$.
\end{proof}

\begin{proof}[Proof of Corollary~\ref{cor:col-lb}]
    Clearly $\chi(G) \ge n/(\alpha(G)n)$. Therefore there must be an identifier assignment such that $A_G(u) \ge \chi(G)$ for some node $u$ in any graph $G$. Now let $T = \F \cup S$ be the disjoint union of the probe view of $A$ at $u$ when run in $G$ and some arbitrary tree $S$ such that $|T| = |G|$, with some arbitrary identifier assignment on $S$. Clearly the probe view of $A$ at $u$ is the same in $G$ and in $T$, and therefore $A_T(u) = A_G(u) \ge \chi(G)$.
\end{proof}

\begin{proof}[Proof of Theorem~\ref{thm:max-cut-apx}]
    The proof proceeds almost exactly as in the proof of Theorem~\ref{thm:adh-mis-apx}, except that we use the family of graphs from part~\ref{thm:lb-graphs-mc} of Theorem~\ref{thm:lb-graphs}. The key property is again that a bipartite double cover of a graph admits a cut in which all edges are cut edges.
\end{proof}

\bibliographystyle{plainnat}
\bibliography{centralised-local}

\begin{thebibliography}{34}
\providecommand{\natexlab}[1]{#1}
\providecommand{\url}[1]{\texttt{#1}}
\expandafter\ifx\csname urlstyle\endcsname\relax
  \providecommand{\doi}[1]{doi: #1}\else
  \providecommand{\doi}{doi: \begingroup \urlstyle{rm}\Url}\fi

\bibitem[Alon(2010)]{alon10constant}
Noga Alon.
\newblock On constant time approximation of parameters of bounded degree
  graphs.
\newblock In Oded Goldreich, editor, \emph{Lecture Notes in Computer Science},
  volume 6390, pages 234--239. Springer Berlin Heidelberg, 2010.
\newblock
  \href{http://dx.doi.org/10.1007/978-3-642-16367-8_14}{\nolinkurl{doi:10.1007/978-3-642-16367-8_14}}.

\bibitem[Alon and Lovett(2013)]{v009a015}
Noga Alon and Shachar Lovett.
\newblock Almost $k$-wise vs. $k$-wise independent permutations, and uniformity
  for general group actions.
\newblock \emph{Theory of Computing}, 9\penalty0 (15):\penalty0 559--577, 2013.
\newblock \doi{10.4086/toc.2013.v009a015}.

\bibitem[Alon et~al.(2012)Alon, Rubinfeld, Vardi, and
  Xie]{alon12space-efficient}
Noga Alon, Ronitt Rubinfeld, Shai Vardi, and Ning Xie.
\newblock Space-efficient local computation algorithms.
\newblock In \emph{Proc.\ 16th Annual ACM-SIAM Symposium on Discrete Algorithms
  (SODA 2012)}, pages 1132--1139. SIAM, 2012.
\newblock
  \href{http://dx.doi.org/10.1137/1.9781611973402}{\nolinkurl{doi:10.1137/1.9781611973402}}.

\bibitem[Bollob{\'a}s(1981)]{bollobas81independence}
B{\'e}la Bollob{\'a}s.
\newblock The independence ratio of regular graphs.
\newblock \emph{Proceedings of the American Mathematical Society}, 83\penalty0
  (2):\penalty0 433--436, 1981.

\bibitem[Campagna et~al.(2013)Campagna, Guo, and Rubinfeld]{CampagnaGR13}
Andrea Campagna, Alan Guo, and Ronitt Rubinfeld.
\newblock Local reconstructors and tolerant testers for connectivity and
  diameter.
\newblock In \emph{Proc. 16th International Workshop, {APPROX} 2013, and 17th
  International Workshop, {RANDOM} 2013}, pages 411--424, 2013.
\newblock \doi{10.1007/978-3-642-40328-6_29}.

\bibitem[Czygrinow et~al.(2008)Czygrinow, Ha{\'n}{\'c}kowiak, and
  Wawrzyniak]{czygrinow2008fast}
Andrzej Czygrinow, Michal Ha{\'n}{\'c}kowiak, and Wojciech Wawrzyniak.
\newblock Fast distributed approximations in planar graphs.
\newblock In \emph{Distributed Computing}, pages 78--92. Springer, 2008.

\bibitem[Even et~al.(2014{\natexlab{a}})Even, Medina, and Ron]{even14best}
Guy Even, Moti Medina, and Dana Ron.
\newblock Best of {T}wo {L}ocal {M}odels: {L}ocal {C}entralized and {L}ocal
  {D}istributed {A}lgorithms, 2014{\natexlab{a}}.
\newblock \href{http://arxiv.org/abs/1402.3796}{\nolinkurl{arXiv:1402.3796}}.

\bibitem[Even et~al.(2014{\natexlab{b}})Even, Medina, and
  Ron]{even14deterministic}
Guy Even, Moti Medina, and Dana Ron.
\newblock Deterministic stateless centralized local algorithms for bounded
  degree graphs.
\newblock In \emph{Proc.\ 22nd European Symposium on Algorithms (ESA 2014)},
  pages 394--405. Springer, 2014{\natexlab{b}}.

\bibitem[Even et~al.(2014{\natexlab{c}})Even, Medina, and
  Ron]{even15distributed}
Guy Even, Moti Medina, and Dana Ron.
\newblock Distributed maximum matching in bounded degree graphs.
\newblock In \emph{Proc.\ 2015 International Conference on Distributed
  Computing and Networking (ICDCN 2015)}, pages 1--19. ACM, 2014{\natexlab{c}}.
\newblock
  \href{http://dx.doi.org/10.1145/2684464.2684469}{\nolinkurl{doi:10.1145/2684464.2684469}}.
  \href{http://arxiv.org/abs/1407.7882}{\nolinkurl{arXiv:1407.7882}}.

\bibitem[Feige et~al.(2015)Feige, Mansour, and Schapire]{feige15learning}
Uriel Feige, Yishay Mansour, and Robert~E Schapire.
\newblock Learning and inference in the presence of corrupted inputs.
\newblock In \emph{Proc.\ 28th Conference on Learning Theory}, volume~40, pages
  637--657. Journal of Machine Learning Research, 2015.
\newblock \url{http://jmlr.org/proceedings/papers/v40/Feige15.html}.

\bibitem[Fich and Ramachandran(1990)]{fich90linkedlists}
Faith~E. Fich and Vijaya Ramachandran.
\newblock Lower bounds for parallel computation on linked structures.
\newblock In \emph{Proc.\ 2nd Annual ACM Symposium on Parallel Algorithms and
  Architectures (SPAA 1990)}, pages 109--116. ACM Press, 1990.
\newblock
  \href{http://dx.doi.org/10.1145/97444.97676}{\nolinkurl{doi:10.1145/97444.97676}}.

\bibitem[Hirvonen et~al.(2014)Hirvonen, Rybicki, Schmid, and
  Suomela]{hirvonen14local-maxcut}
Juho Hirvonen, Joel Rybicki, Stefan Schmid, and Jukka Suomela.
\newblock Large cuts with local algorithms on triangle-free graphs, February
  2014.
\newblock \href{http://arxiv.org/abs/1402.2543}{\nolinkurl{arXiv:1402.2543}}.

\bibitem[Kaplan et~al.(2009)Kaplan, Naor, and Reingold]{knr09}
Eyal Kaplan, Moni Naor, and Omer Reingold.
\newblock Derandomized constructions of k-wise (almost) independent
  permutations.
\newblock \emph{Algorithmica}, 55\penalty0 (1):\penalty0 113--133, 2009.
\newblock \doi{10.1007/s00453-008-9267-y}.

\bibitem[Kuhn et~al.(2004)Kuhn, Moscibroda, and Wattenhofer]{kuhn04what}
Fabian Kuhn, Thomas Moscibroda, and Roger Wattenhofer.
\newblock What cannot be computed locally!
\newblock In \emph{Proc.\ 23rd Annual ACM Symposium on Principles of
  Distributed Computing (PODC 2004)}, pages 300--309. ACM Press, 2004.
\newblock
  \href{http://dx.doi.org/10.1145/1011767.1011811}{\nolinkurl{doi:10.1145/1011767.1011811}}.

\bibitem[Lenzen and Wattenhofer(2008)]{lenzen2008leveraging}
Christoph Lenzen and Roger Wattenhofer.
\newblock Leveraging {L}inial's locality limit.
\newblock In \emph{Distributed Computing}, pages 394--407. Springer, 2008.

\bibitem[Levi et~al.(2014)Levi, Ron, and Rubinfeld]{levi14local}
Reut Levi, Dana Ron, and Ronitt Rubinfeld.
\newblock Local {A}lgorithms for {S}parse {S}panning {G}raphs.
\newblock In \emph{Proc.\ Approximation, Randomization, and Combinatorial
  Optimization.\ Algorithms and Techniques (APPROX/RANDOM 2014)}, pages
  826--842. Schloss Dagstuhl--Leibniz-Zentrum fuer Informatik, 2014.
\newblock
  \href{http://dx.doi.org/10.4230/LIPIcs.APPROX-RANDOM.2014.826}{\nolinkurl{doi:10.4230/LIPIcs.APPROX-RANDOM.2014.826}}.
  \href{http://arxiv.org/abs/1402.3609}{\nolinkurl{arXiv:1402.3609}}.

\bibitem[Levi et~al.(2015{\natexlab{a}})Levi, Moshkovitz, Ron, Rubinfeld, and
  Shapira]{LeviMRRS15}
Reut Levi, Guy Moshkovitz, Dana Ron, Ronitt Rubinfeld, and Asaf Shapira.
\newblock Constructing near spanning trees with few local inspections,
  2015{\natexlab{a}}.
\newblock \href{http://arxiv.org/abs/1502.00413}{\nolinkurl{arXiv:1502.00413}}.

\bibitem[Levi et~al.(2015{\natexlab{b}})Levi, Rubinfeld, and
  Yodpinyanee]{levi15non-constant}
Reut Levi, Ronitt Rubinfeld, and Anak Yodpinyanee.
\newblock Local {C}omputation {A}lgorithms for {G}raphs of {N}on-{C}onstant
  {D}egrees.
\newblock In \emph{Proc.\ 27th ACM on Symposium on Parallelism in Algorithms
  and Architectures (SPAA 2015)}, pages 59--61. ACM Press, 2015{\natexlab{b}}.
\newblock
  \href{http://dx.doi.org/10.1145/2755573.2755615}{\nolinkurl{doi:10.1145/2755573.2755615}}.

\bibitem[Linial(1992)]{linial92locality}
Nathan Linial.
\newblock Locality in distributed graph algorithms.
\newblock \emph{SIAM Journal on Computing}, 21\penalty0 (1):\penalty0 193--201,
  1992.
\newblock
  \href{http://dx.doi.org/10.1137/0221015}{\nolinkurl{doi:10.1137/0221015}}.

\bibitem[Lubotzky et~al.(1988)Lubotzky, Phillips, and
  Sarnak]{lubotzky88ramanujan}
Alexander Lubotzky, Ralph Phillips, and Peter Sarnak.
\newblock Ramanujan graphs.
\newblock \emph{Combinatorica}, 8\penalty0 (3):\penalty0 261--277, 1988.

\bibitem[Lynch(1996)]{lynch96book}
Nancy~A. Lynch.
\newblock \emph{Distributed Algorithms}.
\newblock Morgan Kaufmann Publishers, San Francisco, 1996.

\bibitem[Mansour and Vardi(2013)]{mansour13local}
Yishay Mansour and Shai Vardi.
\newblock A local computation approximation scheme to maximum matching.
\newblock In \emph{Proc.\ 16th International Workshop APPROX 2013, and 17th
  International Workshop RANDOM}, pages 260--273. Springer, 2013.
\newblock
  \href{http://dx.doi.org/10.1007/978-3-642-40328-6_19}{\nolinkurl{doi:10.1007/978-3-642-40328-6_19}}.
  \href{http://arxiv.org/abs/1306.5003}{\nolinkurl{arXiv:1306.5003}}.

\bibitem[Mansour et~al.(2012)Mansour, Rubinstein, Vardi, and
  Xie]{mansour12converting}
Yishay Mansour, Aviad Rubinstein, Shai Vardi, and Ning Xie.
\newblock Converting online algorithms to local computation algorithms.
\newblock In \emph{Proc.\ 39th International Colloquium on Automata, Languages
  and Programming ( ICALP 2012)}, volume 7391, pages 653--664. Springer Berlin
  Heidelberg, 2012.
\newblock
  \href{http://dx.doi.org/10.1007/978-3-642-31594-7_55}{\nolinkurl{doi:10.1007/978-3-642-31594-7_55}}.
  \href{http://arxiv.org/abs/1205.1312}{\nolinkurl{arXiv:1205.1312}}.

\bibitem[Mohar and Poljak(1990)]{mohar90eigenvalue}
Bojan Mohar and Svatopluk Poljak.
\newblock Eigenvalue and the max-cut problem.
\newblock \emph{Czechoslovak Mathematical Journal}, 40\penalty0 (2):\penalty0
  343--352, 1990.
\newblock \url{http://dml.cz/dmlcz/102386}.

\bibitem[Morgenstern(1994)]{morgenstern94existence}
Moshe Morgenstern.
\newblock Existence and explicit constructions of {$q+1$} regular {R}amanujan
  graphs for every prime power {$q$}.
\newblock \emph{Journal of Combinatorial Theory, Series B}, 62\penalty0
  (1):\penalty0 44--62, 1994.
\newblock
  \href{http://dx.doi.org/10.1006/jctb.1994.1054}{\nolinkurl{doi:10.1006/jctb.1994.1054}}.

\bibitem[Naor and Stockmeyer(1995)]{naor95what}
Moni Naor and Larry Stockmeyer.
\newblock What can be computed locally?
\newblock \emph{SIAM Journal on Computing}, 24\penalty0 (6):\penalty0
  1259--1277, 1995.
\newblock
  \href{http://dx.doi.org/10.1137/S0097539793254571}{\nolinkurl{doi:10.1137/S0097539793254571}}.

\bibitem[Nguyen and Onak(2008)]{nguyen08constant-time}
Huy~N. Nguyen and Krzysztof Onak.
\newblock Constant-time approximation algorithms via local improvements.
\newblock In \emph{Proc.\ 49th Annual IEEE Symposium on Foundations of Computer
  Science (FOCS 2008)}, pages 327--336. IEEE Computer Society Press, 2008.
\newblock
  \href{http://dx.doi.org/10.1109/FOCS.2008.81}{\nolinkurl{doi:10.1109/FOCS.2008.81}}.

\bibitem[Parnas and Ron(2007)]{parnas07approximating}
Michal Parnas and Dana Ron.
\newblock Approximating the minimum vertex cover in sublinear time and a
  connection to distributed algorithms.
\newblock \emph{Theoretical Computer Science}, 381\penalty0 (1--3):\penalty0
  183--196, 2007.
\newblock
  \href{http://dx.doi.org/10.1016/j.tcs.2007.04.040}{\nolinkurl{doi:10.1016/j.tcs.2007.04.040}}.

\bibitem[Peleg(2000)]{peleg00distributed}
David Peleg.
\newblock \emph{Distributed Computing: A Locality-Sensitive Approach}.
\newblock SIAM Monographs on Discrete Mathematics and Applications. Society for
  Industrial and Applied Mathematics, Philadelphia, 2000.

\bibitem[Reingold and Vardi(2014)]{reingold14new}
Omer Reingold and Shai Vardi.
\newblock New {T}echniques and {T}ighter {B}ounds for {L}ocal {C}omputation
  {A}lgorithms, 2014.
\newblock \href{http://arxiv.org/abs/1404.5398}{\nolinkurl{arXiv:1404.5398}}.

\bibitem[Rubinfeld et~al.(2011)Rubinfeld, Tamir, Vardi, and
  Xie]{rubinfeld11fast}
Ronitt Rubinfeld, Gil Tamir, Shai Vardi, and Ning Xie.
\newblock Fast local computation algorithms.
\newblock In \emph{Proc.\ Innovations in Computer Science (ICS 2011)}, 2011.
\newblock \href{http://arxiv.org/abs/1104.1377}{\nolinkurl{arXiv:1104.1377}}.

\bibitem[Shearer(1983)]{shearer83note}
James~B. Shearer.
\newblock A note on the independence number of triangle-free graphs.
\newblock \emph{Discrete Mathematics}, 46\penalty0 (1):\penalty0 83--87, 1983.
\newblock
  \href{http://dx.doi.org/10.1016/0012-365X(83)90273-X}{\nolinkurl{doi:10.1016/0012-365X(83)90273-X}}.

\bibitem[Shearer(1992)]{shearer92bipartite}
James~B. Shearer.
\newblock A note on bipartite subgraphs of triangle-free graphs.
\newblock \emph{Random Structures \& Algorithms}, 3\penalty0 (2):\penalty0
  223--226, 1992.
\newblock
  \href{http://dx.doi.org/10.1002/rsa.3240030211}{\nolinkurl{doi:10.1002/rsa.3240030211}}.

\bibitem[Suomela(2013)]{suomela13survey}
Jukka Suomela.
\newblock Survey of local algorithms.
\newblock \emph{ACM Computing Surveys}, 45\penalty0 (2):\penalty0 24:1--40,
  2013.
\newblock
  \href{http://dx.doi.org/10.1145/2431211.2431223}{\nolinkurl{doi:10.1145/2431211.2431223}}.
  \url{http://www.cs.helsinki.fi/local-survey/}.

\end{thebibliography}

\end{document}